\documentclass[10pt]{siamltex}

\usepackage[english]{babel}
\usepackage{amssymb,amsfonts,amsmath,latexsym}
\usepackage{graphicx}

\numberwithin{equation}{section}

\newcommand{\bm}[1]{\mathbf{#1}}

\newcommand{\R}{{\mathbb R}}
\newcommand{\RR}{{\mathbb R}}
\newcommand{\MM}{{\mathbb M}}

\DeclareMathAlphabet{\pcal}{OMS}{zplm}{m}{n}

\begin{document}

\title{Block Matrix Formulations for Evolving Networks}


\author{Caterina Fenu\thanks{Department of Computer Science,
University of Pisa, Largo Bruno Pontecorvo, 3, 56127 Pisa, Italy. E-mail: \texttt{kate.fenu@unica.it}.} 
\and
Desmond J.~Higham\thanks{Department of Mathematics and Statistics,
University of Strathclyde, Glasgow, U.K. E-mail:
\texttt{d.j.higham@strath.ac.uk}.}
Research supported by EPSRC/RCUK Established Career fellowship 
EP/M00158X/1 and a Royal Society/Wolfson Research Merit Award. }

\pagestyle{myheadings}
\markboth{C.~Fenu and D.~J.~Higham}{Block Matrix Formulations for Evolving Networks}


\maketitle

\begin{abstract} 
Many types of pairwise interaction take the form of a fixed set of nodes 
with edges that appear and disappear over time.
In the case of discrete-time evolution, the resulting evolving network may be 
represented by a time-ordered sequence of adjacency matrices. 
We consider here the issue of representing the system as a single, higher dimensional block matrix,
built from the individual time-slices.
 We focus on the task of computing network centrality measures.
From a modeling perspective, we show that there is a suitable block formulation that 
allows us to recover dynamic centrality measures respecting time's arrow. 
From a computational perspective, we show that the new block formulation leads to the design of 
more effective numerical algorithms. 
\end{abstract}



\section{Introduction}\label{sec:intro}
A \emph{multilayer network}, also called \emph{network of networks} \cite{multirevb,multirevjcn}, is a graph
where connections are formed within and between well-defined slices, each of which is itself a network.   
In this case it is natural to regard the connectivity structure as a three-dimensional tensor.
We focus here on a specific type of multilayering where each slice represents a 
time point.   
More precisely, let $\{G^{[k]}\}_{k=1}^M = \left(V,\{E^{[k]}\}_{k=1}^M \right)$ 
be a sequence of unweighted graphs evolving in discrete time. Here,  the set of nodes $V$,
with $|V| = n$, is fixed and the evolution in time is given by the change in the set of edges,
$E^{[k]}$.
With this notation, given the 
ordered sequence of time points $\{t_k\}_{k=1}^M$, the network at time $t_k$ is represented by its $n \times n$ adjacency matrix $A^{[k]}$. As usual for unweighted networks, the $(i,j)$th entry of $A^{[k]}$ 
equals $1$ if there is an edge from node $i$ to node $j$ at time $t_k$, and $0$ otherwise.  
This type of connectivity structure arises naturally in many types of human interaction. 
For example, within a given population, we may record physical interactions, phone calls, text messages, emails, social media contacts  or correlations between behavior such as energy usage or on-line shopping; see 
\cite{HS12} for an overview.

Although we may regard 
$\{  
A^{[k]}
\}_{k=1}^M
$
as a three dimensional tensor, 
we emphasize 
that, in this context,  the third dimension is very different from the first two.
Typical quantities of interest are invariant to the ordering the nodes---we may consistently 
 permute the rows and columns of each $A^{[k]}$, or, equivalently, we may relabel the nodes, 
without affecting our conclusions.
However, for most purposes, it is not appropriate to reorder the time points. 
This raises a question that motivates the work presented here: to what extent 
can we rely on ideas from  the generic 
multi-layer/tensor viewpoint
when studying evolving networks?
More specifically, focusing on the idea of
\emph{flattening} a tensor into a single, larger, two-dimensional matrix (also known as 
\emph{reshaping}, \emph{unfolding} or  \emph{matricizing}) 
\cite{GVL3,KoBa09,SVL12}, how do we express an evolving network as a single, large, block matrix?
We address this question in the context of computing node centrality.

The material is organized as follows.  
In section~\ref{sec:cent}, we review a class of centrality measures based on the 
concept of dynamic walks.
Section~\ref{sec:block} then presents a flattening of the adjacency matrices from which these  
centrality measures can be recovered using standard matrix functions.
Methods for the computation of the centrality measures using the new block approach are presented in~\ref{sec:comptask}.  
Computational experiments on synthetic and voice call data are described in 
section~\ref{sec:comp}.  Within these tests, we also study the 
supra-centrality matrix formulation from 
\cite{TMCPM15}.
Final conclusions are given in section~\ref{sec:conc}.

\section{Centrality} \label{sec:cent}
In this section we review the concepts of time-dependent centrality measures 
from~\cite{AH15,GHsr12, GHPE11}  
and introduce a new connection between them. 
Centrality measures are widely used for identifying influential players in a network.
Many such measures arose within the field of social network analysis, motivated either explicitly or implicitly from the idea that the network nodes communicate, or pass information, along the edges; see, for example,
\cite{B05,Free79}.
In this way, centrality quantifies a sense in which a node takes part in  traversals.  
Quoting from \cite{BE2006}
\lq\lq 
All measures of centrality assess a node's involvement in the walk structure of a network.\rq\rq

For the time-dependent links that we consider here, it has been pointed out by several  authors that any type of message-passing (or disease-passing) basis for centrality should
account for the time-ordering of the interactions; see, for example, \cite{HS12}.   
If $X$ meets $Y$ today and $Y$ meets $Z$ tomorrow, then the path $X \to Y \to Z$ makes sense
from a message-passing point of view, but not
$Z \to X \to Y$. Traversals must respect the arrow of time.

In \cite{GHPE11}, as a means to develop a time-dependent centrality measure,  
the authors introduced the notion of a \emph{dynamic walk} as follows.
\begin{definition}
\label{def:dw}
A \emph{dynamic walk} of length $w$ from node $i_1$ to node $i_{w+1}$ consists of a sequence of edges $i_1 \rightarrow i_2$, $i_2 \rightarrow i_3$, \dots, $i_w \rightarrow i_{w+1}$ and a nondecreasing sequence of times $t_{r_1} \leq t_{r_2} \leq \cdots \leq t_{r_w}$ such that $A^{[r_m]}_{i_m,i_{m+1}} \ne 0$. 
\end{definition}

This definition was used to define the \emph{dynamic communicability matrix} 
\begin{equation}\label{dcm}
\mathcal{Q}^{[j]} = \left(I - aA^{[1]}\right)^{-1} \left(I - aA^{[2]}\right)^{-1} \cdots \left(I - aA^{[j]}\right)^{-1} = \prod_{s=1}^j \left(I - aA^{[s]}\right)^{-1}.
\end{equation}
We assume henceforth that the parameter $a$ satisfies 
$a < 1/\max_s{\rho(A^{[s]})}$, with $\rho(A^{[s]})$ denoting the spectral radius of the matrix $A^{[s]}$. Each resolvent in \eqref{dcm} may then be expanded as 
$$
\left(I - aA^{[s]}\right)^{-1} = \sum_{k=0}^{\infty}a^k \left(A^{[s]}\right)^k.
$$  
In view of this, $\left(\mathcal{Q}^{[j]}\right)_{ij}$ can be seen as a weighted sum of the number of dynamic walks from $i$ to $j$ using the ordered sequence $\{A^{[k]}\}_{k=1}^M$, in which the count for walks of length $w$ is scaled by $a^w$. The overall ability of nodes to 
broadcast or receive information in this sense is given by the row and column sums 
\begin{equation}\label{br}
C^{{\rm broadcast}} = \mathcal{Q}^{[j]} \bm{1} \qquad \text{and} \qquad C^{{\rm receive}} = {\mathcal{Q}^{[j]}}^T \bm{1},
\end{equation}
respectively, where $\bm{1}$ is the vector of all ones. See \cite{GHPE11} for a more detailed explanation.  Numerical tests in \cite{GHPE11} showed that these broadcast 
and receive centralities
are generally very different from the measures that arise when we  ignore time-dependency and consider only the aggregate adjacency matrix $\sum_{k=1}^{M} A^{[k]}$,    
 and subsequent work in \cite{LMGAOH13}
 showed that they were better able to match the views of 
social media experts when applied to Twitter data. 

Motivated by the treatment of static networks in
\cite{EH}, the authors in \cite{AH15} used the dynamic communicability matrix idea to introduce two kinds of dynamic betweenness: the \emph{nodal betweenness} of a node and the \emph{temporal betweenness} of a time point.

Let $\bar{A}^{[k]}_r$ denote the matrix obtained from $A^{[k]}$ by removing all the edges involving node $r$, that is, $\bar{A}^{[k]}_r = A^{[k]} - E_r^{[k]}$, where $E_r^{[k]}$ has nonzero elements only in row and column $r$, which coincide with those of $A^{[k]}$. Then, the matrix
$$
\bar{\mathcal{Q}}^{[M]}_r = \prod_{s=1}^M \left(I - a\bar{A}^{[s]}_r\right)^{-1}
$$
quantifies the ability of nodes to  communicate without using node $r$.
The \emph{nodal betweenness of node $r$}, \cite{AH15}, is defined as
\begin{equation}\label{nb}
{\rm NB}_r := \frac{1}{(n-1)^2 - (n-1)} \mathop{\sum\sum}_{i \neq j \neq r}\frac{(\mathcal{Q}^{[M]})_{ij} - (\bar{\mathcal{Q}}^{[M]}_r)_{ij}}{(\mathcal{Q}^{[M]})_{ij}}.
\end{equation}
This measure quantifies the relative decrease in information exchange 
when node $r$ is removed from the network. 

Let $\{\widehat{A}^{[k,q]}\}_{k=1}^M$ denote the adjacency matrix sequence obtained on  replacing $A^{[q]}$ with $0$, that is
$$
\widehat{A}^{[k,q]} = (1-\delta_{kq})A^{[k]},
$$  
where $\delta_{kq}$ is the Kronecker delta. Then, the matrix
$$
\widehat{\mathcal{Q}}^{[M,q]} = \prod_{s=1}^M \left(I - a\widehat{A}^{[s,q]}\right)^{-1} = \mathop{\prod_{s=1}}_{s \neq q}^M \left(I - a A^{[s]}\right)^{-1}
$$
describes how well nodes interchange information without using the connections at time $q$.
The \emph{temporal betweenness}, \cite{AH15}, of time point $q$ is defined as
\begin{equation}\label{tb}
{\rm TB}^{[M,q]} := \frac{1}{(n-1)^2 - (n-1)} \mathop{\sum\sum}_{i \neq j}\frac{(\mathcal{Q}^{[M]})_{ij} - (\widehat{\mathcal{Q}}^{[M,q]})_{ij}}{(\mathcal{Q}^{[M]})_{ij}}.
\end{equation}
We refer to \cite{AH15} for further details and illustrative examples involving these measures, and we note that in practical use the matrices $\mathcal{Q}^{[M]}$, $\bar{\mathcal{Q}}^{[M]}_r$ and $\widehat{\mathcal{Q}}^{[M,q]}$ should be properly scaled in order to prevent over/underflows.

For future convenience, we extend the notation to allow for walks that start and finish at arbitrary 
time points.
 Let us denote by $\mathcal{Q}^{[i,j]}$ the dynamic communicability matrix obtained by multiplying the resolvents corresponding to the ordered sequence $\{A^{[s]}\}_{s=i}^j, 1 \leq i \leq j \leq M$, that is, 
\begin{equation}\label{dcm2}
\mathcal{Q}^{[i,j]} = \prod_{s=i}^j \left(I - aA^{[s]}\right)^{-1} = \left(I - aA^{[i]}\right)^{-1} \cdots \left(I - aA^{[j]}\right)^{-1}.
\end{equation}
With this notation we can quantify broadcast and receive centralities over any subinterval.
In general, we may use 
\begin{equation}\label{br2}
C_{{\rm broadcast}}^{[i,j]} = \mathcal{Q}^{[i,j]} \bm{1} \qquad \text{and} \qquad C_{{\rm receive}}^{[i,j]} = {\mathcal{Q}^{[i,j]}}^T \bm{1},
\end{equation}
to quantify the ability of a node to spread or receive information, respectively, taking into account the evolution of the network between $t_i$ and $t_j$.

In many applications, such as the spread of rumors or disease, recent walks 
are more important than those that started a long time ago. For this reason, the authors in \cite{GHsr12} introduced the \emph{running dynamic communicability matrix}, 
 $\mathcal{S}^{[j]}$, obtained recursively, starting from $\mathcal{S}^{[0]} = 0$, as   
\begin{equation}\label{rdcm}
\mathcal{S}^{[j]} = \left(I + e^{-b\Delta t_j}\mathcal{S}^{[j-1]}\right)\left(I-aA^{[j]}\right)^{-1} - I, \qquad j = 1, \dots, M,
\end{equation}
where $\Delta t_j = t_j - t_{j-1}$.
In this recurrence the parameter $a$ is used to penalize long walks and the parameter $b$ to filter out old activity.
Overall,  $\mathcal{S}^{[j]}$
maintains walk counts that are scaled in terms of length $w$ by $a^w$ and chronological 
age $t$ by $e^{-bt}$.  
Running versions of the  broadcast and receive 
communicabilities are then given by the row/column sums of the matrix $\mathcal{S}^{[j]}$, that is
\begin{equation}\label{rbr}
\mathcal{S}^{[j]} \bm{1} \qquad \text{and} \qquad {\mathcal{S}^{[j]}}^T \bm{1}.
\end{equation}

For use in the next section, 
the following lemma points out a connection between the running dynamic communicability matrix $\mathcal{S}^{[j]}$ in \eqref{rdcm} and the dynamic communicability matrices $\mathcal{Q}^{[i,j]}$ in \eqref{dcm2}.

\begin{lemma}\label{thm:rdcm}
For the 
running dynamic communicability matrix 
$\mathcal{S}^{[j]}$ in \eqref{rdcm} we have 
$$
\mathcal{S}^{[j]} = \sum_{i=1}^j \left(1-e^{-b\Delta t_i} \right) e^{-b \sum_{\ell = i+1}^j \Delta t_{\ell}} \mathcal{Q}^{[i,j]} - I,
$$
where $\Delta t_1 = \infty$.
\end{lemma}
\begin{proof}
The proof uses induction.
For $j = 1$, we have 
$$
\mathcal{S}^{[1]} = \left(I + e^{-b\Delta t_1}\mathcal{S}^{[0]}\right)\left(I-aA^{[1]}\right)^{-1} - I = \mathcal{Q}^{[1,1]} - I.
$$
Suppose that the identity is valid for $j = k-1$. We will then show it is valid also for $j = k$.
We have
$$
\begin{aligned}
\mathcal{S}^{[k]} &= \left(I + e^{-b\Delta t_k}\mathcal{S}^{[k-1]}\right)\mathcal{Q}^{[k,k]} - I \\
&= \mathcal{Q}^{[k,k]} + e^{-b\Delta t_k} \left[ \sum_{i=1}^{k-1} \left(1-e^{-b\Delta t_i} \right) e^{-b \sum_{\ell = i+1}^{k-1} \Delta t_{\ell}} \mathcal{Q}^{[i,k-1]} - I\right]  \mathcal{Q}^{[k,k]} - I \\
&= \sum_{i=1}^{k-1} \left(1-e^{-b\Delta t_i} \right) e^{-b \sum_{\ell = i+1}^{k-1} \Delta t_{\ell}} e^{-b\Delta t_k} \mathcal{Q}^{[i,k-1]}\mathcal{Q}^{[k,k]} + \left(1-e^{-b\Delta t_k} \right) \mathcal{Q}^{[k,k]} - I \\
&= \sum_{i=1}^{k-1} \left(1-e^{-b\Delta t_i} \right) e^{-b \sum_{\ell = i+1}^{k} \Delta t_{\ell}} \mathcal{Q}^{[i,k]} + \left(1-e^{-b\Delta t_k} \right) \mathcal{Q}^{[k,k]} - I \\
&= \sum_{i=1}^{k} \left(1-e^{-b\Delta t_i} \right) e^{-b \sum_{\ell = i+1}^{k} \Delta t_{\ell}} \mathcal{Q}^{[i,k]}  - I,
\end{aligned}
$$ 
where we used the fact that $\mathcal{Q}^{[i,k-1]}\mathcal{Q}^{[k,k]} = \mathcal{Q}^{[i,k]}$.
\end{proof}


\section{Block matrix formulations} \label{sec:block}

Our aim now is to study   
block-matrix
representations of the data 
 $\{ A^{[k]} \}_{k=1}^{M}$
that transform the network sequence into an \lq\lq equivalent\rq\rq\ large, static network
with adjacency matrix of dimension $M  n$.
We have two main requirements for such a representation.  
\begin{itemize}
 \item 
  We would like to be able to interpret this static network in terms of the interactions represented
by the original data.
 \item  We would like to be able to recover the dynamic centrality measures
discussed in the previous section by applying standard matrix functions to this larger network.   
 \end{itemize}

In the case of general multi-layer networks, the authors in 
\cite{SRCFGB13} introduce an \emph{influence matrix} $W \in \RR^{ M \times M}$ such that  
$w_{ij} \ge 0 $ measures the influence of layer $j$ on layer $i$.
They then study node centrality via the $M n$ by $M n$ matrix 
\[
\begin{bmatrix}
 w_{11} A^{[1]} &  w_{12} A^{[2] }     &  \ldots      &          &   w_{1M} A^{[M] }  \\
  w_{21} A^{[1]}      &  w_{22} A^{[2]} &    \ldots   &           &  w_{2M} A^{[M]}  \\
   \vdots     &   \vdots      & \ddots &       &    \vdots \\
         &         &        &        &   \\
      w_{M1} A^{[1]}  &   w_{M2} A^{[2] }      &   \hdots     &           & w_{MM} A^{[M]}
\end{bmatrix}.
\]

In our specific context, where (a) the layers represent time slices that have a natural ordering, 
and (b) centrality concepts are motivated from traversals around the network,  
 this formulation appears to add little value. 
If each $M \times M$ block represents a time slice, then the existence of an edge at one 
time slice should not influence the propensity for traversal \emph{within} some other time slice. 
Hence, only the simple block-diagonal version ($w_{ij} = 0$ for all $i \neq j$) makes intuitive
sense in our context. 

Returning  to Definition~\ref{def:dw}, we note that a dynamic walk may use any number of edges within a time slice and may then \emph{wait} until a later time slice and continue the traversal.
For dynamic communicability defined via 
(\ref{dcm}),  
within each time  slice, use of an edge penalizes the walk count by $a$ and 
moving from one time slice to the next carries no penalty.
For the more general running measure based on
(\ref{rdcm}), waiting until the next time slice costs a factor $e^{-b \Delta t_j}$.
We may capture this type of weighted count by introducing a link \emph{from a node 
at one time slice to the equivalent node at the next time slice}; that is, by adding 
identity matrices along the first superdiagonal to obtain
 $B \in \R^{Mn \times Mn}$ defined as 
\begin{equation}
B := \begin{bmatrix}
\alpha A^{[1]} &    \beta_2 I    &        &           &   \\
        & \alpha A^{[2]} &    \beta_3 I   &           &   \\
        &         & \ddots &   \ddots  &   \\
        &         &        & \alpha A^{[M-1]} & \beta_M I \\
        &         &        &           & \alpha A^{[M]}
\end{bmatrix}.
\label{bigmat}
\end{equation}
Here $ \{ \beta_{\ell} \}_{\ell = 2, M}$ and $\alpha$ are parameters.
The next theorem confirms that this structure captures the required communicabilities when 
$\alpha = a$ and $\beta_{\ell} \equiv 1$ or $\beta_{\ell} = e^{-b\Delta t_\ell}$ for the two cases.

\begin{theorem}\label{thm:mf}
The dynamic communicability matrices $\mathcal{Q}^{[i,j]}$ in \eqref{dcm2} and the running dynamic communicability matrices $\mathcal{S}^{[j]}$ in \eqref{rdcm} can be computed 
by applying the function  $f(x) = (1-x)^{-1}$ to the matrix $B$ in \eqref{bigmat}.
\end{theorem}

\begin{proof}
It is straightforward to show that the $k$th power of the matrix $B$ has the form
$$
B^k = \scriptscriptstyle{\begin{bmatrix}
\alpha^k h_k\left(A^{[1]}\right) & \beta_2\alpha^{k-1}  h_{k-1}(A^{[1]}, A^{[2]}) & \cdots & \prod_{\ell=2}^M\beta_{\ell} \alpha^{k-r}  h_{k-r}(A^{[1]}, \dots, A^{[M]})  \\
                   &    \alpha^k h_k\left(A^{[2]}\right)    & \ddots & \vdots \\
                   &                           &       \ddots         &  \beta_M \alpha^{k-1}  h_{k-1}(A^{[M-1]}, A^{[M]}) \\
                   &                               &  & \alpha^k h_k\left(A^{[M]}\right)          
\end{bmatrix}},
$$
where 
$$
h_k(x_1,x_2,\dots,x_n) = \sum_{l_1+l_2+\cdots l_n = k}{x_1}^{l_1}{x_2}^{l_2}\cdots{x_n}^{l_n}
$$ 
is the \emph{complete homogeneous symmetric polynomial} of degree $k$, $r = M-1$ and $h_k(\cdot, \cdot,\dots, \cdot) = 0$ if $k<0$.

In general, denoting block $(i,j)$, $(i,j = 1,\dots,M)$ of $B^k$ 
by $[B^k]_{ij}$, we have that 
$$
[B^k]_{ij} = \beta^{[i+1,j]}\alpha^{k+i-j} h_{k+i-j}(A^{[i]},\dots,A^{[j]}), \qquad i \leq j,
$$
where $\beta^{[i+1,j]}$ denotes the scalar $\prod_{\ell=i+1}^j \beta_{\ell}$.

The matrix-valued function $f(B) = \sum_{k=0}^\infty B^k$ has blocks 
$$
\left[f(B)\right]_{ij} = \beta^{[i+1,j]} \sum_{k=0}^\infty \alpha^{k+i-j} h_{k+i-j}(A^{[i]},\dots,A^{[j]}) = \beta^{[i+1,j]}\prod_{\ell = i}^j (I-\alpha A^{[\ell]})^{-1}. 
$$
Hence,  the dynamic communicability matrices $\mathcal{Q}^{[i,j]}$ can be obtained from the block $\left[f(B)\right]_{ij}$ setting $\beta_{\ell} = 1$, $\ell = 2, \dots, M$ and $\alpha = a$.

The running dynamic communicability matrices $\mathcal{S}^{[j]}$ are obtained starting from the blocks on the $j$th block-column. In particular, setting $\beta_{\ell} = e^{-b\Delta t_\ell}$,
$\alpha = a$ and $D = \diag(1-e^{-b\Delta t_1},\dots,1-e^{-b\Delta t_M}) \otimes I_n$, 
we have 
$$
\begin{aligned}
\left[\sum_{i=1}^j \left[D f(B)\right]_{ij}\right] - I &= \sum_{i=1}^j \left(1-e^{-b\Delta t_i} \right) \beta^{[i+1,j]} \mathcal{Q}^{[i,j]} - I \\
&= \sum_{i=1}^j \left(1-e^{-b\Delta t_i} \right) e^{-b \sum_{\ell = i+1}^j \Delta t_{\ell}} \mathcal{Q}^{[i,j]} - I.
\end{aligned}
$$ 
The statement now follows from Lemma~\ref{thm:rdcm}.
\end{proof}

Similar statements apply to the betweenness measures in (\ref{nb}) and (\ref{tb}).

\begin{theorem}\label{thm:fm2}
Let $\bar{A}^{[k]}_r$ denote the matrix obtained from $A^{[k]}$ by removing all the edges involving node $r$ and let $\{\widehat{A}^{[k,q]}\}_{k=1}^M$ be the adjacency matrix sequence obtained replacing $A^{[q]}$ with $0$. Then, for $f(x) = (1-x)^{-1}$, 
$$
\begin{aligned}
{\rm NB}_r &= \frac{1}{(n-1)^2 - (n-1)} \mathop{\sum\sum}_{i \neq j \neq r}\frac{\left[f(B)\right]_{1M}^{ij} - \left[f(\bar{B}_r)\right]_{1M}^{ij}}{\left[f(B)\right]_{1M}^{ij}}, \\
{\rm TB}^{[M,q]} &= \frac{1}{(n-1)^2 - (n-1)} \mathop{\sum\sum}_{i \neq j}\frac{\left[(f(B)\right]_{1M}^{ij} - [(f(\widehat{B}^{[q]})]_{1M}^{ij}}{\left[f(B)\right]_{1M}^{ij}},
\end{aligned}
$$
where $\bar{B_r}$ and $\widehat{B}^{[q]}$ are given by
$$
\begin{aligned}
\bar{B_r} &= \begin{bmatrix}
\alpha \bar{A}^{[1]}_r &     I    &        &           &   \\
        & \alpha \bar{A}^{[2]}_r &     I   &           &   \\
        &         & \ddots &   \ddots  &   \\
        &         &        & \alpha \bar{A}^{[M-1]}_r &  I \\
        &         &        &           & \alpha \bar{A}^{[M]}_r
\end{bmatrix} \\
\widehat{B}^{[q]} &= \begin{bmatrix}
\alpha A^{[1]} &     I    &        &           & &           &  \\
        & \alpha A^{[2]} &     I   &           & &           &  \\
        &         & \ddots &   \ddots  &   &           &\\
        &         &        & \alpha A^{[q-1]} &  I &           &\\
        &         &        &           & \alpha A^{[q+1]} &   I     &        \\
        &         &        &           &                  & \ddots & \ddots \\
        &         &        &           &       &           &   \alpha A^{[M-1]} & I\\
        &         &        &           &        &          &  & \alpha A^{[M]} \\
\end{bmatrix}
\end{aligned}
$$
and $\left[(f(B)\right]_{1M}^{ij}$ denotes the $(i,j)$th element of the $(1,M)$th block of the matrix $f(B)$.
\end{theorem}
\begin{proof}
A proof follows along the same lines as that of Theorem~\ref{thm:mf}.
\end{proof}

\subsection{Another formulation} \label{sec:supra}

An alternative block matrix formulation that was specifically designed for 
evolving networks appears in \cite{TMCPM15}.
Those authors define the  
 \emph{supra-centrality matrix}
to have the general form
\begin{equation}
  \MM
= 
 \left[
 \begin{array}{ccccc}
    \epsilon M^{[1]}   & I & 0 & \ldots  & 0\\
    I   &  \epsilon M^{[2]}   & I & \ddots  & \vdots \\   
   0 & I   &  \epsilon M^{[3]}    & \ddots & 0  \\
  \vdots & \ddots & \ddots &  \ddots & 	I	 \\
 0 & \ldots & 0 &  I & \epsilon M^{[M]} 
 \end{array}
 \right].
\label{eq:MM}
\end{equation}
Here, $M^{[k]}$ is an $n$ by $n$ centrality matrix based on 
$A^{[k]}$. 
The authors use the simple choice $M^{[k]} \equiv A^{[k]}$ to illustrate the idea, but mention that other static centrality functions could be used, such as the Katz 
\cite{Katz53} resolvent-based 
version
 $M^{[k]} \equiv (I - \alpha A^{[k]})^{-1}$  (which is the single time-point case of 
(\ref{dcm})).
It is then proposed in 
\cite{TMCPM15}
   to apply a standard static network centrality algorithm to the supra-centrality matrix 
 $\MM$. 
The parameter $\epsilon$
in (\ref{eq:MM}) 
 is included to account for the fact that the identity matrices represent
\lq\lq between layer\rq\rq\ connections 
that are inherently different from the \lq\lq within layer\rq\rq\ weights arising from the
network data.
A key element of (\ref{eq:MM})  is the appearance of identity matrices in the super- 
\emph{and sub-} block-diagonal positions. 
If the overall centrality measure applied to  $\MM$ is motivated by monitoring traversals  around the large, static,  
$Mn $ by $Mn$ network, then, because of the identity matrices on the sub-diagonal blocks, some of these traversals will be travelling \emph{backwards in time} with respect to the original time-stamped data.
Similarly,
suppose that
each 
$A^{[k]}$ is symmetric, so that edges are undirected in each time slice. 
Then with 
$M^{[k]} \equiv A^{[k]}$ or 
$M^{[k]} \equiv (I - \alpha A^{[k]})^{-1}$,
 we see that  $\MM$ is symmetric. 
 However, from the simple example mentioned in section~\ref{sec:cent}, where 
 $X$ meets $Y$ today and $Y$ meets $Z$ tomorrow, 
we can see that time's arrow introduces asymmetry, even when the individual
interactions are symmetric. 
 In Section~\ref{sec:comp}, we will perform some illustrative tests that compare results from 
(\ref{bigmat}) and 
(\ref{eq:MM}). 

\section{Computational tasks}\label{sec:comptask}

In this section we provide some useful ways to deal with the computation of the quantities defined in section~\ref{sec:cent}. In particular, we focus on the case of running broadcast and receive communicabilities given by~\eqref{rbr} since, as pointed out in~\cite{GHsr12}, the running dynamic communicability matrices 
$\mathcal{S}^{[j]}$ exhibit a typical ``fill in'' behavior.   



From Theorem~\ref{thm:mf}, setting $\beta_{\ell} = e^{-b\Delta t_\ell}$ and $\alpha = a$, we obtain
\begin{equation}\label{bil}
\begin{aligned}
{\mathcal{S}^{[j]}}\bm{1}_{\rm n} &= (\bm{d} \otimes I_{\rm n})^Tf(B)(\bm{e}_j \otimes \bm{1}_{\rm M}) - \bm{1}_{\rm n} \\
{\mathcal{S}^{[j]}}^T\bm{1}_{\rm n} &= (\bm{e}_j \otimes I_{\rm n})^Tf(B)^T(\bm{d} \otimes \bm{1}_{\rm n}) - \bm{1}_{\rm n},
\end{aligned}
\end{equation}
where $\bm{d} = [1,1-\beta_{2},\dots,1-\beta_{M}]^T$, $\bm{1}_{\rm M}$ and $\bm{1}_{\rm n}$ are vectors of all ones in $\R^{M}$ and $\R^{n}$, respectively. 

The formulation opens up two computational approaches: 
one based on bilinear forms involving matrix functions and one that focuses on the solution of a sparse linear system. In the following we will compare these methods in terms of execution time and accuracy.

\subsection{Use of quadrature formulas} \label{sec:quadform}

We are interested in computation of quantities of the form 
$$
\bm{u}^T f(B) \bm{v}, \qquad  \bm{u}, \bm{v} \in \R^{Mn},
$$ 
with $\bm{u},\bm{v}$ unit vectors and $f(B) = (I-B)^{-1}$ nonsymmetric. 

In particular,
$(\bm{u} = \bm{d} \otimes \bm{e}_i, \bm{v} = \bm{e}_j \otimes \bm{1}_{\rm M})$ and $(\bm{u} = \bm{e}_j \otimes \bm{e}_i, \bm{v} = \bm{d} \otimes \bm{1}_{\rm n})$, $i = 1, \dots, n$, for the broadcast and receive running communicabilities of node $i$, respectively.

In this case, since both the vectors $\bm{u}$ and $\bm{v}$ and the matrix $B$ are very sparse, the probability of breakdown during the computation is high. For this reason, it is convenient to add a dense vector to each initial vector (see~\cite{BDY}) and resort to a block algorithm. In particular, we use the nonsymmetric block Lanczos algorithm~\cite{BookGM} and pairs of block Gauss and anti-Gauss quadrature rules~\cite{CRS,FMRR,La}.

If $U = [\bm{u} \; \bm{1}]$ and $V = [\bm{v} \; \bm{1}]$, then we want to approximate the quantities 
$$
U^T f(B) V, \qquad  U, V \in \R^{Mn \times 2}. 
$$ 

The nonsymmetric block Lanczos algorithm applied to the 
matrix B with initial blocks $U_1 = U$ and $V_1 = V$ yields, after $\ell$ steps, the decompositions
\begin{equation*}
\begin{aligned}
B\left[U_1,\dots,U_\ell\right]
				&=\left[U_1,\dots,U_\ell\right]J_\ell+U_{\ell+1}\Gamma_\ell 					\bm{E}_\ell^T,\\
B^T\left[V_1,\dots,V_\ell\right]	&=\left[V_1,\dots,V_\ell\right]J_\ell^T+V_{\ell+1}\Delta_\ell \bm{E}_\ell^T,
\end{aligned}
\end{equation*}
where $J_\ell$ is the matrix 
\begin{equation}\label{JL}
J_\ell=\left(
\begin{array}{ccccc}
\Omega_1&\Delta_1^T&&&\\
\Gamma_1&\Omega_2&\Delta_2^T&&\\

&\ddots&\ddots&\ddots&\\
&&\Gamma_{\ell-2}&\Omega_{\ell-1}&\Delta_{\ell-1}^T\\
&&&\Gamma_{\ell-1}&\Omega_\ell

\end{array}
\right)\in\R^{2\ell\times 2\ell},
\end{equation}
and $\bm{E}_k = \bm{e}_k^T \otimes I_2$, for $k=1,2,\ldots,\ell$ are $2\times (2\ell)$ block matrices which contain $2\times 2$ zero blocks everywhere, except for the $k$th block, which coincides with the identity matrix $I_2$.
The $\ell$-block nonsymmetric Gauss
quadrature rule ${\pcal G}_\ell$ can then be expressed as
\begin{equation*}
{\pcal G}_\ell = \bm{E}_1^Tg(J_\ell) \bm{E}_1.
\end{equation*}

As shown in \cite{FMRR}, the
$(\ell+1)$-block nonsymmetric anti-Gauss rule can be computed in terms of 
the matrix $\widetilde{J}_{\ell+1}$ as
\begin{equation*}
{\pcal H}_{\ell+1}= \bm{E}_1^Tg(\widetilde{J}_{\ell+1}) \bm{E}_1,
\end{equation*}
where 
$$
\widetilde{J}_{\ell+1}=\left(\begin{array}{r|c}
&   \\
 \makebox[50pt][c]{\Large{$J_{\ell}$}}  &  \\ 
&   \sqrt{2}\Delta_{\ell}^T \\ \hline
  \sqrt{2}\Gamma_{\ell} &\Omega_{\ell+1}
\end{array}\right)\in\R^{2(\ell+1)\times 2(\ell+1)}.
$$

Since $f(x) = (1-x)^{-1}$, with $|x| < 1$, is analytic in a 
simply connected domain whose boundary encloses the spectrum of $B$ but is not close to it (see~\cite{FMRR} for details), the arithmetic mean 
\begin{equation}\label{FL}
F_\ell = \frac{1}{2}\left({\pcal G}_\ell + {\pcal H}_{\ell+1}\right)
\end{equation} 
between Gauss and anti-Gauss quadrature rules can be used as an approximation of the matrix-valued expression $U^T f(B) V$.

\subsection{Resolution of a sparse linear system}\label{sec:linsys}

Using the same notation as in subsection~\ref{sec:quadform}, we need to compute the quantities 
$$
\bm{u}^T(I-B)^{-1}\bm{v}, \qquad  \bm{u}, \bm{v} \in \R^{Mn}.
$$
This can be done by solving the sparse linear system $(I-B)\bm{x} = \bm{v}$ and then computing the scalar product $\bm{u}^T\bm{x}$. 

The linear system can be solved either directly or iteratively. The peculiarity of the block formulation allows us to have at hand a regular matrix splitting. In fact, we have $I \geq 0$, $B \geq 0$ and $\rho(B) < 1$. Therefore, the iterative method
$$
\bm{x}^{(k+1)} = B\bm{x}^{(k)} + \bm{v},
$$
with given starting vector $\bm{x}^{(0)}$, converges to the solution $\bm{x}$.

Another classical approach to solve linear systems is the \emph{LSQR method} that makes use of the Golub-Kahan algorithm.
After $\ell$ steps of this method with starting vector $q_1 = \bm{v}$, the solution $\bm{x}^{(\ell)} \in \R^{Mn}$ is defined as 
$$
\bm{x}^{(\ell)} = P_\ell \bm{y}^{(\ell)} = \beta_1 P_\ell C_{\ell+1,\ell}^{\dagger}\bm{e}_1,
$$ 
that is, $\bm{y}^{(\ell)} \in \R^{\ell}$ is the solution of the least squares problem
$$
\min{\| C_{\ell+1,\ell} \bm{y}^{(\ell)} - \beta_1 \bm{e}_1 \|_2},
$$
where $C_{\ell+1,\ell}$ is computed via the Golub-Kahan algorithm, $\beta_1 = \|\bm{v}\|$ and $\bm{e}_1$ is the first vector of the canonical base of size $\ell+1$.

\section{Computational tests}\label{sec:comp}

In this section we perform some numerical tests in order to judge effectiveness, 
 both from the computational and modeling points of view.
First, we compare the methods described in the previous section against the original approach presented in~\cite{GHsr12}, that is, by using the recursive formula~\eqref{rdcm}.
Then, we show the relevance of the upper triangular block formulation compared with 
the supra-centrality matrix approach given in~\cite{TMCPM15}.

\subsection{Computation using the new block approach}\label{sec:newblock}

As a first set of experiments we compare different ways to deal with the computation of the quantities defined by~\eqref{rbr}. In particular we focus on the computation of the running broadcast communicabilities $\mathcal{S}^{[j]}\bm{1}$. 
We recall that the computation of these communicabilities at a given time step can not be recovered using the same information at a previous time step, that is, the update of the running broadcast communicabilities $\mathcal{S}^{[j]}\bm{1}$ needs the computation of the whole matrix $\mathcal{S}^{[j]}$. 

We analyze the following methods:
\begin{description}
\item[original] is the original approach presented in~\cite{GHsr12}. In particular, we use the \mbox{\textsf{mldivide}} MATLAB function to compute the inverse of the matrix $\left(I-aA^{[j]}\right)$ in recursion~\eqref{rdcm}.
\item[quadrules] is the approach based on the Gauss and anti-Gauss quadrature rules described in subsection~\ref{sec:quadform}. More precisely, we perform as many steps of the block nonsymmetric Lanczos algorithm as necessary to obtain a relative distance
$$
\frac{\|{\pcal G}_\ell - {\pcal H}_{\ell+1}\|_{\max}}{\|F_\ell\|_{\max}}, \qquad \text{with} \quad \|X\|_{\max}:=\max_{1\le i,j\le k}|X_{ij}|, 
$$  
less than $10^{-3}$.
\item[linsolv] is the first procedure described in subsection~\ref{sec:linsys} in which we solve the big linear system $(I-B)\bm{x} = \bm{v}$ using the \textsf{mldivide} MATLAB function.    
\item[iterative] is the second method proposed in subsection~\ref{sec:linsys}, namely the iterative approach based on the regular matrix splitting $I - B$. We perform as many iterations as necessary to reach a relative accuracy of $10^{-3}$ on the difference between two consecutive approximations.
\item[lsqr] is the method based on the solution of the linear system obtained by using the \textsf{lsqr} MATLAB function. In particular, we set the tolerance to $10^{-3}$.
\end{description}

We want to test the performance of the methods when the size of the matrix $B$ in~\eqref{bigmat} increases. This can be done various ways. As a first approach, in order to simulate the evolution on a given set of nodes, we
independently sample $M$ times from the same static network model with a fixed number of nodes $n$. This was be done in MATLAB using the package CONTEST by Taylor and Higham~\cite{TH}. As a second approach, we generate the $M$ matrices by using the 
evolving network model proposed and analyzed in~\cite{GHP12}.
Here, the network sequence corresponds to the sample path of a discrete time Markov chain,
 and hence the adjacency matrices are correlated  over time.
 All computations were carried out with MATLAB version 9.0 (R2016a) 64-bit 
for Linux, in double precision arithmetic, on an Intel(R) Xeon(R) computer with 32 Gb RAM.

Table~\ref{tab:1} shows the results obtained for the scale free random graph model generated using the \textsf{pref} function of the CONTEST toolbox, which implements a preferential attachment model. We set a fixed number of nodes $n = 10^3$ and $M$ goes from $10$ to $100$, in order to test the performance of the methods when the number of time steps is increasing. 
The table displays the time required to compute the running broadcast communicabilities of the $n$ nodes of the evolving network and the absolute error 
$$
\|\bm{x} - \tilde{\bm{x}} \|_{\infty} = \max_i |x_i - \tilde{x}_i|,
$$  
where $\tilde{\bm{x}}$ is the approximation and $\bm{x}$ is the vector computed with the original approach.

\begin{table}[htb]
\small\setlength{\tabcolsep}{4.2pt}
\begin{center}
\caption{Execution time and absolute error for the computation of the running broadcast communicabilities with $n=1000$ and $M=10, \dots, 100$. The network model is obtained from the \textsf{pref} function of the CONTEST toolbox}
\label{tab:1}
\begin{tabular}{c|c|cc|cc|cc|c}
    & \textsf{original}    & \multicolumn{2}{c}{\textsf{quadrules}} & \multicolumn{2}{|c|}{\textsf{linsolv}} & \multicolumn{2}{c|}{\textsf{iterative}} & \textsf{lsqr}\\
$M$ &   time   & time &  err. & time &  err. & time &  err. & time \\ \hline
10  & 1.68e+01 & 2.24e+01 & 5.75e-03 & 3.18e-01 & 1.30e-14 & 3.15e-02 & 8.87e-04 & 9.55e-02\\
20  & 3.46e+01 & 4.36e+01 & 3.34e-03 & 4.05e+00 & 1.49e-14 & 1.17e-02 & 1.06e-03 & 3.87e-02\\
30  & 5.35e+01 & 7.42e+01 & 3.44e-03 & 1.42e+01 & 1.58e-14 & 1.49e-02 & 1.07e-03 & 4.81e-02\\
40  & 7.06e+01 & 1.06e+02 & 6.23e-03 & 2.77e+01 & 1.59e-14 & 1.12e-02 & 1.05e-03 & 6.56e-02\\
50  & 8.94e+01 & 1.22e+02 & 4.81e-03 & 7.07e+01 & 1.65e-14 & 1.28e-02 & 1.04e-03 & 6.31e-02\\
60  & 1.16e+02 & 1.52e+02 & 2.52e-03 & 6.08e+01 & 1.42e-14 & 1.40e-02 & 1.01e-03 & 6.76e-02\\
70  & 1.33e+02 & 1.95e+02 & 1.85e-03 & 5.66e+01 & 1.53e-14 & 1.62e-02 & 1.02e-03 & 8.01e-02\\
80  & 1.45e+02 & 2.44e+02 & 3.30e-03 & 1.04e+02 & 1.37e-14 & 1.62e-02 & 9.22e-04 & 9.85e-02\\
90  & 1.73e+02 & 2.63e+02 & 6.73e-03 & 5.71e+01 & 1.59e-14 & 1.83e-02 & 1.08e-03 & 1.27e-01\\
100 & 1.94e+02 & 2.48e+02 & 3.94e-03 & 1.53e+02 & 1.38e-14 & 1.94e-02 & 1.06e-03 & 1.49e-01\\
\end{tabular}
\end{center}
\end{table}

The results clearly show that the quadrature rules based on the block Lanczos algorithm do not improve the performance of the original method, while both the methods based on the resolution of the big linear system work very well. In particular, since we are interested in the rank of the nodes rather than the value of the index, the iterative method gives good results and it is very fast. It is also worth noting that the \textsf{lsqr} method, whose results coincide with those obtained by using the \textsf{mldivide} function, is very fast and very accurate.

\begin{table}[htb]
\small\setlength{\tabcolsep}{4.2pt}
\begin{center}
\caption{Execution time and absolute error for the computation of the running broadcast communicabilities with $n=1000$ and $M=10, \dots, 100$. The network model is obtained from the \textsf{renga} function of the CONTEST toolbox}
\label{tab:2}
\begin{tabular}{c|c|cc|cc|cc|c}
& \textsf{original} & \multicolumn{2}{c}{\textsf{quadrules}} & \multicolumn{2}{|c|}{\textsf{linsolv}} & \multicolumn{2}{c|}{\textsf{iterative}} & \textsf{lsqr}\\
$M$ &   time   & time &  err. & time &  err. & time &  err. & time \\ \hline
10 & 1.61e+01 & 6.14e+01 & 3.25e-02 & 2.84e-01 & 8.22e-15 & 1.80e-01 & 5.63e-03 & 1.75e-01\\ 
20 & 3.34e+01 & 1.27e+02 & 2.18e-02 & 8.90e-01 & 7.11e-15 & 2.48e-02 & 8.89e-03 & 1.71e-01\\ 
30 & 5.24e+01 & 1.83e+02 & 2.31e-02 & 1.54e+00 & 8.44e-15 & 2.70e-02 & 1.10e-02 & 2.41e-01\\ 
40 & 6.87e+01 & 2.23e+02 & 3.93e-02 & 2.34e+00 & 8.22e-15 & 2.58e-02 & 1.04e-02 & 3.14e-01\\ 
50 & 8.90e+01 & 2.96e+02 & 1.26e-01 & 2.76e+00 & 7.55e-15 & 2.78e-02 & 8.21e-03 & 4.45e-01\\ 
60 & 1.07e+02 & 3.40e+02 & 2.18e-02 & 3.60e+00 & 7.11e-15 & 2.61e-02 & 1.20e-02 & 4.38e-01\\ 
70 & 1.24e+02 & 4.02e+02 & 1.41e-02 & 4.24e+00 & 7.77e-15 & 2.99e-02 & 9.52e-03 & 5.30e-01\\ 
80 & 1.39e+02 & 4.55e+02 & 8.73e-02 & 4.94e+00 & 7.33e-15 & 3.02e-02 & 1.02e-02 & 5.90e-01\\ 
90 & 1.60e+02 & 6.01e+02 & 6.77e-02 & 5.66e+00 & 7.77e-15 & 3.06e-02 & 1.10e-02 & 6.54e-01\\ 
100 & 1.74e+02 & 7.09e+02 & 1.66e-02 & 6.51e+00 & 6.66e-15 & 3.31e-02 & 8.91e-03 & 7.09e-01\\ 
\end{tabular}
\end{center}
\end{table}

To investigate the behavior of the methods proposed with respect to the network model, we performed the same computation as in Table~\ref{tab:1} using a range dependent random graph generated using the \textsf{renga} function of the CONTEST toolbox. Table~\ref{tab:2} shows the results obtained setting $n = 10^3$ and varying $M$ from $10$ to $100$. The results show that the block quadrature rule method and the iterative resolution of the big linear system are slower and the error is greater than that obtained from \textsf{pref}. However, the performance of the direct solution of the linear system is faster and gives a small error. Again, the \textsf{lsqr} method is the best among those proposed.

\begin{table}[htb]
\small\setlength{\tabcolsep}{4pt}
\begin{center}
\caption{Execution time and absolute error for the computation of the running broadcast communicabilities with $M=10$ and $n=1000, \dots, 10000$. The network model is obtained from the \textsf{pref} function of the CONTEST toolbox}
\label{tab:3}
\begin{tabular}{c|c|cc|cc|cc|c}
& \textsf{original} & \multicolumn{2}{c}{\textsf{quadrules}} & \multicolumn{2}{|c|}{\textsf{linsolv}} & \multicolumn{2}{c|}{\textsf{iterative}} & \textsf{lsqr}\\
$n$ &   time   & time &  err. & time &  err. & time &  err. & time \\ \hline
1000 & 1.79e+01 & 2.15e+01 & 5.75e-03 & 3.19e-01 & 1.30e-14 & 2.63e-02 & 8.87e-04 & 9.40e-02\\ 
2000 & 1.37e+02 & 7.17e+01 & 1.11e-02 & 1.14e+00 & 2.42e-14 & 1.68e-02 & 8.11e-04 & 5.24e-02\\ 
3000 & 4.72e+02 & 1.66e+02 & 1.24e-02 & 2.43e+00 & 3.70e-14 & 2.55e-02 & 8.33e-04 & 7.58e-02\\ 
4000 & 1.11e+03 & 2.71e+02 & 5.44e-03 & 4.47e+00 & 5.46e-14 & 2.83e-02 & 7.17e-04 & 1.07e-01\\ 
5000 & 2.18e+03 & 4.02e+02 & 2.90e-02 & 8.05e+00 & 6.32e-14 & 3.49e-02 & 7.46e-04 & 1.31e-01\\ 
6000 & 3.92e+03 & 5.64e+02 & 3.28e-03 & 1.15e+01 & 8.00e-14 & 4.20e-02 & 7.71e-04 & 1.61e-01\\ 
7000 & 5.88e+03 & 6.80e+02 & 2.15e-02 & 1.61e+01 & 8.35e-14 & 4.87e-02 & 7.35e-04 & 1.97e-01\\ 
8000 & 9.17e+03 & 9.19e+02 & 1.30e-02 & 2.32e+01 & 9.15e-14 & 5.67e-02 & 7.34e-04 & 2.27e-01\\ 
9000 & 1.31e+04 & 1.43e+03 & 8.18e-02 & 3.39e+01 & 1.27e-13 & 6.52e-02 & 7.30e-04 & 2.77e-01\\ 
10000 & 1.91e+04 & 1.78e+03 & 1.08e-02 & 3.95e+01 & 9.34e-14 & 7.39e-02 & 7.17e-04 & 3.16e-01\\
\end{tabular}
\end{center}
\end{table}

We now investigate the behavior of the methods when the number of time steps is fixed and the size of the network increases. Table~\ref{tab:3}--~\ref{tab:4} show the results obtained when $M = 10$ and $n$ goes from $10^{3}$ to $10^{4}$ with respect to the \textsf{pref} model and the \textsf{renga} model, respectively. It is clear that the original method strongly depends on the size of the matrices, while the methods proposed here work quickly and effectively. It is worth noting that we need to wait more than one hour to obtain the value of the index for a network with $6000$ nodes or more by using the original approach. 
We see that the method based on the iterative solution of the linear system is not only the fastest among the five, but tolerates very well the change of dimension, making this approach a very good method to deal with large networks.

\begin{table}[htb]
\small\setlength{\tabcolsep}{3.5pt}
\begin{center}
\caption{Execution time and absolute error for the computation of the running broadcast communicabilities with $M=10$ and $n=1000, \dots, 10000$. The network model is obtained from the \textsf{renga} function of the CONTEST toolbox}
\label{tab:4}
\begin{tabular}{c|c|cc|cc|cc|c}
& \textsf{original} & \multicolumn{2}{c}{\textsf{quadrules}} & \multicolumn{2}{|c|}{\textsf{linsolv}} & \multicolumn{2}{c|}{\textsf{iterative}} & \textsf{lsqr}\\
$n$ &   time   & time &  err. & time &  err. & time &  err. & time\\ \hline
1000 & 1.59e+01 & 5.87e+01 & 3.25e-02 & 2.72e-01 & 8.22e-15 & 2.69e-02 & 5.63e-03 & 1.92e-01\\ 
2000 & 1.29e+02 & 2.23e+02 & 3.17e-02 & 5.73e-01 & 8.22e-15 & 3.31e-02 & 7.50e-03 & 2.11e-01\\ 
3000 & 4.70e+02 & 4.68e+02 & 1.16e+00 & 8.73e-01 & 8.88e-15 & 4.77e-02 & 7.29e-03 & 3.39e-01\\ 
4000 & 1.11e+03 & 8.21e+02 & 1.34e-01 & 1.23e+00 & 8.88e-15 & 6.40e-02 & 6.55e-03 & 4.67e-01\\ 
5000 & 2.76e+03 & 1.35e+03 & 7.80e-02 & 1.55e+00 & 8.66e-15 & 8.03e-02 & 6.53e-03 & 6.76e-01\\ 
6000 & 6.11e+03 & 1.99e+03 & 9.68e-02 & 2.22e+00 & 7.33e-15 & 1.15e-01 & 5.59e-03 & 8.20e-01\\ 
7000 & 1.12e+04 & 2.62e+03 & 9.17e-02 & 2.24e+00 & 8.44e-15 & 1.13e-01 & 6.40e-03 & 8.07e-01\\ 
8000 & 1.90e+04 & 4.28e+03 & 1.05e-01 & 2.80e+00 & 9.10e-15 & 1.56e-01 & 6.14e-03 & 9.24e-01\\ 
9000 & 3.17e+04 & 4.73e+03 & 2.73e-01 & 3.20e+00 & 1.04e-14 & 2.07e-01 & 5.98e-03 & 1.05e+00\\ 
10000 & 4.59e+04 & 6.14e+03 & 3.32e-01 & 3.59e+00 & 9.99e-15 & 1.95e-01 & 5.73e-03 & 1.18e+00\\
\end{tabular}
\end{center}
\end{table}

\begin{table}[htb]
\small\setlength{\tabcolsep}{3.3pt}
\begin{center}
\caption{Execution time and relative error for the computation of the running broadcast communicabilities with $M=10$ and $n=1000, \dots, 10000$. The network sequence is obtained from the \emph{triadic closure model}.}
\label{tab:5}
\begin{tabular}{c|c|cc|cc|cc|c}
& \textsf{original} & \multicolumn{2}{c}{\textsf{quadrules}} & \multicolumn{2}{|c|}{\textsf{linsolv}} & \multicolumn{2}{c|}{\textsf{iterative}} & \textsf{lsqr}\\
$n$ &   time   & time &  err. & time &  err. & time &  err. & time\\ \hline
1000 & 2.04e+01 & 1.98e+02 & 1.12e-01 & 1.99e+00 & 1.18e-15 & 5.55e-01 & 1.12e-03 & 7.86e-01 \\ 
2000 & 1.70e+02 & 1.08e+03 & 3.95e-01 & 8.67e+00 & 2.08e-15 & 8.69e-01 & 2.32e-03 & 3.53e+00 \\ 
3000 & 6.02e+02 & 3.35e+03 & 5.72e-02 & 2.36e+01 & 1.68e-15 & 2.45e+00 & 2.95e-03 & 8.08e+00 \\ 
4000 & 1.40e+03 & 7.13e+03 & 1.76e-01 & 4.27e+01 & 2.41e-15 & 5.43e+00 & 3.92e-03 & 1.58e+01 \\ 
5000 & 2.82e+03 & 1.32e+04 & 7.41e-01 & 7.33e+01 & 1.99e-15 & 8.47e+00 & 3.88e-03 & 2.52e+01 \\ 
6000 & 5.15e+03 & 2.35e+04 & 8.78e-01 & 1.10e+02 & 2.39e-15 & 1.30e+01 & 3.87e-03 & 3.98e+01 \\ 
7000 & 7.74e+03 & 3.36e+04 & 9.98e-01 & 1.65e+02 & 2.47e-15 & 1.90e+01 & 4.10e-03 & 5.54e+01 \\ 
8000 & 1.27e+04 & 5.91e+04 & 9.94e-01 & 2.22e+02 & 2.60e-15 & 2.72e+01 & 5.02e-03 & 7.57e+01 \\ 
9000 & 1.70e+04 & 7.35e+04 & 1.00e+00 & 3.19e+02 & 3.52e-15 & 3.34e+01 & 4.72e-03 & 1.22e+02 \\ 
10000 & 2.31e+04 & 1.05e+05 & 1.00e+00 & 1.38e+03 & 3.03e-15 & 4.64e+01 & 4.90e-03 & 1.88e+02 \\ 
\end{tabular}
\end{center}
\end{table}

\begin{table}[htb]
\small\setlength{\tabcolsep}{3.8pt}
\begin{center}
\caption{Execution time and relative error for the computation of the running broadcast communicabilities with $n=1000$ and $M=10, \dots, 100$. The network sequence is obtained from the \emph{triadic closure model}.}
\label{tab:6}
\begin{tabular}{c|c|cc|cc|cc|c}
& \textsf{original} & \multicolumn{2}{c}{\textsf{quadrules}} & \multicolumn{2}{|c|}{\textsf{linsolv}} & \multicolumn{2}{c|}{\textsf{iterative}} & \textsf{lsqr}\\
$M$ &   time   & time &  err. & time &  err. & time &  err. & time\\ \hline
10 & 1.95e+01 & 3.20e+02 & 9.88e-02 & 9.26e-01 & 1.16e-15 & 2.63e+00 & 3.69e-03 & 2.41e+00 \\ 
20 & 3.95e+01 & 1.60e+03 & 2.73e-01 & 2.54e+00 & 8.98e-16 & 5.32e+00 & 5.53e-03 & 3.46e+00 \\ 
30 & 5.83e+01 & 3.99e+03 & NaN & 5.24e+00 & 2.27e-15 & 8.16e+00 & 6.55e-03 & 4.88e+00 \\ 
40 & 8.27e+01 & 1.09e+04 & NaN & 8.82e+00 & 2.04e-15 & 1.08e+01 & 7.53e-03 & 6.28e+00 \\ 
50 & 9.91e+01 & 1.43e+04 & NaN & 1.37e+01 & 4.44e-15 & 1.34e+01 & 9.21e-03 & 8.03e+00 \\ 
60 & 1.31e+02 & 2.23e+04 & NaN & 1.90e+01 & 3.86e-15 & 1.57e+01 & 9.75e-03 & 9.25e+00 \\ 
70 & 1.46e+02 & 2.77e+04 & NaN & 2.50e+01 & 1.84e-15 & 1.86e+01 & 1.10e-02 & 1.19e+01 \\ 
80 & 1.63e+02 & 3.68e+04 & NaN & 3.29e+01 & 1.31e-15 & 2.21e+01 & 1.17e-02 & 1.27e+01 \\ 
90 & 1.82e+02 & 4.80e+04 & NaN & 4.24e+01 & 1.43e-14 & 2.49e+01 & 1.23e-02 & 1.44e+01 \\ 
100 & 2.17e+02 & 6.31e+04 & NaN & 4.59e+01 & 3.26e-15 & 2.72e+01 & 6.69e-02 & 1.68e+01 \\  
\end{tabular}
\end{center}
\end{table}

As a second set of numerical experiments we make use of the \emph{triadic closure model} developed in~\cite{GHP12}. Starting from an 
Erd\"os-R\'enyi network model~\cite{ER59} with a given edge density, we generate a sequence of $M$ matrices in which the network at time point $k+1$ is built starting from the network at the previous time point. In particular, the expected value of $A^{[k+1]}$ given $A^{[k]}$ is 
$$
\mathcal{F}(A^{[k]}) = (1 - \tilde{\omega}) A^{[k]} + (\bm{1}^T\bm{1} - A^{[k]}) \circ (\delta\bm{1}^T\bm{1} + \epsilon (A^{[k]})^2),
$$
where $\tilde{\omega} \in (0,1)$ is the death rate, $\delta\bm{1} + \epsilon (A^{[k]})^2$ is the birth rate, with $0 < \delta \ll 1$ and $0 < \epsilon(n-2) < 1 - \delta$, and $\bm{1}$ is the vector of all ones.  This model is based on the social science hypothesis that \lq\lq friends of friends\rq\rq\ tend to become friends; that is, new edges are more likely between pairs of nodes that are separated by many paths of length two. 

Tables~\ref{tab:5}--\ref{tab:6} display the results obtained by setting $\tilde{\omega} = \delta = 20/n^2$ and $\epsilon = 5/n^2$, where $A^{[1]}$ is an Erd\"os-R\'enyi network model with an edge density of $0.1$ and $0.3$, respectively. We report the execution time and the relative error between the approximate solution and the solution obtained by using the original approach.
The results show that the computation based on the quadrature rules is ineffective in this 
case---Nan indicates that convergence was not attained. This can be explained by taking into account the sparsity level of the matrices involved in the computation, which does not allow us to gain advantage from the use of matrix-vector products. On the contrary, the behavior of the methods based on the solution of the linear system is satisfactory, but again the sparsity level influences the performance of the iterative method based on the matrix splitting. This fact is more evident in Table~\ref{tab:6}, where we obtain comparable results from the methods based on the solution of the linear system.

The computed examples point out the effectiveness of the new block formulation
relative to the original approach, especially when the dimension of the 
individual matrices is high. It is clear that \emph{inverse-free} algorithms based on matrix-vector products are efficient for very sparse networks. Moreover, these kinds of algorithms work well on modern machines, since the computation can be fully parallelized.  

\subsection{New block formulation vs. supra-centrality matrix}

Having used the new block formulation (\ref{bigmat})
to develop efficient computational  strategies, we now compare the relevance of the associated 
centrality measures with those of the supra-centrality version 
(\ref{eq:MM}). 
We first conduct a numerical test based  on a synthetic time-dependent network.
We generate the network in such a way that one node has a temporal connectivity 
pattern that allows it to initiate a disproportionate number of traversals.
We note that this type of hierarchical pattern of interactions has been found, either 
explicitly  or implicitly, in empirical studies of online behavior.
For example, 
in  the context of online forums, 
Graham and Wright
\cite{GW14}
singled out   
\emph{agenda-setters}, who are responsible for 
 new thread creation, and thereby
influence \emph{subsequent} interactions, writing that 
\lq\lq 
  The inclusion of agenda-setting reflects our view that
influence is not limited to the volume of posts alone.\rq\rq\  
Huffaker et al.\  
\cite{Huff09} 
discovered hierarchy within 
the use of
chat features in a  
 Massive
Multiplayer Online (MMO) role-playing game, and found that 
in general \lq\lq players send
messages to higher-level experts.\rq\rq\ 
It is therefore useful to have centrality  tools that can discover and quantify this type of 
influence in the time-dependent setting.

To build a simple data set, 
we use $n =  200$ nodes and $M = 4$ time levels.
We begin by setting each $A^{[k]}$ to be an independent, directed random graph where the probability of an edge from node $i$ to node $j$ at time $k$ is given by $4/n$, independently of 
$i$, $j$ and $k$.   
In this way each node has an expected out degree of $4$ at each time level and there is
no structure to the interactions. 
We then remove all edges that emanate from node $1$.
Finally, we  repeat the following procedure 16 times:
\begin{itemize}
 \item at time level $k=1$ connect node $1$ to a uniformly chosen node, $n_2$, 
 \item at time level $k=2$ connect node $n_2$ to a uniformly chosen node, $n_3$,
 \item at time level $k=3$ connect node $n_3$ to a uniformly chosen node, $n_4$,
\item at time level $k=4$ connect node $n_4$ to a uniformly chosen node, $n_5$.
\end{itemize}
In this way, node $1$ is given $16$ edges that are guaranteed to have a follow-on effect 
in terms of dynamic walks around the network.
In the above  construction, target nodes $n_1, n_2, \ldots$ 
are chosen uniformly and independently 
across $1,2,\ldots,n$, and in a final processing step, repeated edges within a time level and self loops are deleted.     

The upper left picture in 
Figure~\ref{fig:test1} scatter plots, for each node, the aggregate out degree on the horizontal axis against the dynamic broadcast communicability, as given by the row sums of a
normalized 
version, 
$\mathcal{S}^{[M]}/
\|
\mathcal{S}^{[M]} \|_2
$, 
of the 
dynamic communicability matrix
from
(\ref{rdcm}). 
Here we used $\alpha = 0.9/\rho^\star$, where 
$\rho^\star = 4.2$ is the maximum spectral radius over the time levels, and 
$b = 1$ with $\Delta t = 1$.
In the picture, node $1$ is highlighted with a star symbol.
We see that despite having only a typical aggregate out degree, node $1$  
produces by far the highest communicability score, which reflects the fact that its edges have a 
knock-on effect through time.
 In the upper right picture in
 Figure~\ref{fig:test1}, we repeat the test with the time levels taken in reverse order. 
In this case, the built-in dynamic walks \emph{finish} at node $1$, rather than starting there, and   
 the benefit of these walks is now shared more evenly among the randomly chosen initial and intermediary nodes. The performance of node $1$ is now more 
compatible with its aggregate out degree and hence the dynamic communicability measure 
does not highlight    any special structure. 

In the same way, 
the lower left and lower right pictures in Figure~\ref{fig:test1} scatter plot aggregate out degree against
a nodal centrality measure based on 
the   
supra-centrality matrix
 (\ref{eq:MM}) for the original and time-reversed data, respectively.
Here, we used static  Katz centrality  matrices along the diagonal,
so 
$
M^{[k]} = ( I - \alpha A^{[k]})^{-1} 
$, 
with $\alpha$ chosen as in the first two experiments.
To maintain compatibility we also used $\epsilon = e$.
For our overall centrality measure, we again used the Katz resolvent, that is 
$   ( I - {\widehat \alpha}  \MM)^{-1} $,
with 
$\widehat \alpha$ chosen to be a factor $0.9$ times the reciprocal of the 
spectral radius of $ \MM$.  
To obtain a single measure for each node, 
we used the \emph{marginal node centrality} measure defined in 
\cite{TMCPM15}.
We see that this type of centrality calculation, which does not maintain the time ordering, fails 
to highlight the role of node $1$.

\begin{figure}[t]
\begin{center}
\includegraphics[width=0.8\textwidth]{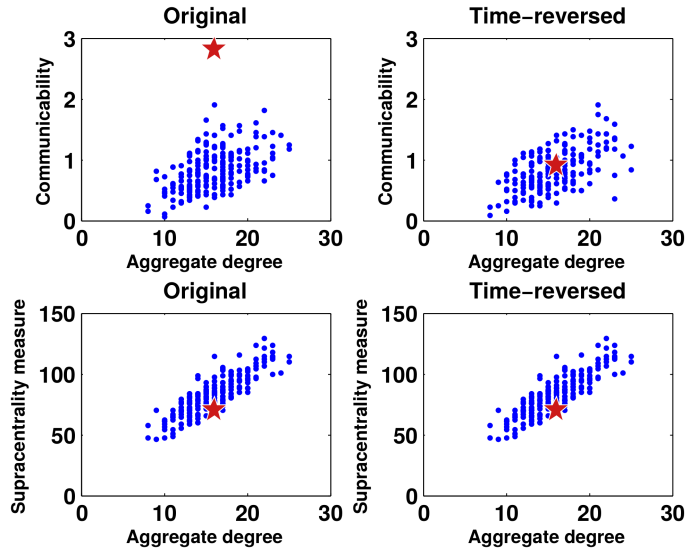}
\caption{Node centrality scatter plots for a synthetic network. The special node, 
$1$, is marked with a star. 
In each case the horizontal axis shows aggregate out degree.
Upper and lower left pictures use dynamic communicability and 
supra-centrality-based marginal node centrality, respectively, for the vertical axis.
The upper and lower right pictures repeat the experiment with the data in reverse time order. 
}
\label{fig:test1}
\end{center}
\end{figure}

Next we
use a set of simulated voice call data from 
the IEEE VAST 2008 Challenge 
\cite{DBLP:conf/ieeevast/GrinsteinPLOSW08}. 
This dynamic data set is 
designed to represent interactions between 
a controversial socio-political movement, and it incorporates some  
unusual temporal activity.
The data involves 
$n = 400$ cell phone users, giving 
a complete set of time stamped  pairwise calls between them. Each call 
is logged via the send and receive nodes, a 
start time  
and a duration in seconds. 
Among the extra information supplied by the 
competition designers was the strong suggestion that 
one node acts as the \lq\lq ringleader\rq\rq\ within a key inner circle.
Based on 
analyses submitted by challenge teams, we believe that 
this ringleader has ID 200, and the rest of 
the inner circle  consists of four nodes with IDs 1, 2, 3 and 5.
Further details can 
currently be found at  
\begin{verbatim}
http://www.cs.umd.edu/hcil/VASTchallenge08/index.htm
\end{verbatim}

This data was studied in terms of temporal centrality in \cite{GHrsoc}, 
where it was shown that a continuous-time  version of broadcast centrality can 
identify the key players, even though they are not the dominant users in terms of aggregate call 
time. 

For our discrete-time experiment, 
we use 30 minute time windows over days 1 to 6. So, the symmetric 
adjacency matrix $A^{[k]}$ records whether nodes $i$ and $j$ spent any time 
interacting in the $k$th 30 minute time window.
To compute the broadcast centrality  
(\ref{rbr}) 
we took $\alpha$ to be a factor of  $0.9$ times the reciprocal of the 
maximum spectral radius of the
$\rho(A^{[k]})$  over $k$, and $b = 0.1$ with $\Delta t = 1$.
As in the previous experiment, we chose comparable parameters for the  
supra-centrality
matrix  
 (\ref{eq:MM}).
Here, we used static  Katz centrality  matrices along the diagonal,
so 
$
M^{[k]} = ( I - \alpha A^{[k]})^{-1} 
$, 
with the same $\alpha$ and with $\epsilon = e^{b}$.
The overall centrality measure was then based on row sums of the Katz resolvent,   
$   ( I - {\widehat \alpha}  \MM)^{-1} $,
with 
$\widehat \alpha$ chosen to be a factor $0.9$ times the reciprocal of the 
spectral radius of $ \MM$. 

Figure~\ref{fig:vast} scatter plots the broadcast centrality against the
supra-centrality based measure. 
Here the ringleader node is marked with a red diamond and the other four 
inner circle nodes are marked with a red five-point star. 
We see that both centrality measures highlight two particular inner circle members 
and all four inner circle members appear in the top seven of both centrality rankings.
However, for the ringleader, marked with a diamond, broadcast centrality 
ranks the node 3rd, whereas 
the supra-centrality measure places this node 48th (out of the $400$ nodes).
We conclude that, in this experiment, the time-respecting measure is better able 
to discover the importance of the ringleader node.

\begin{figure}[t]
\begin{center}
\includegraphics[width=0.8\textwidth]{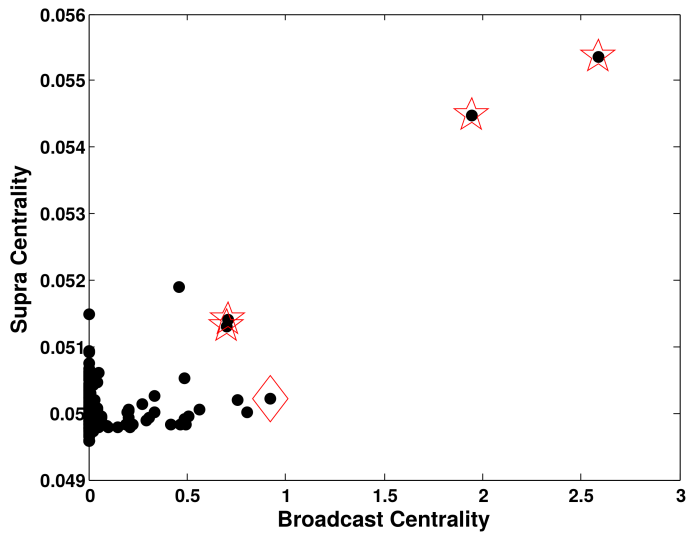}
\caption{Scatter plots of broadcast and supra-centrality based centrality for a 
$400$ node 
voice call network.
Here, five particular nodes are known to be influential.
The ringleader node is marked with a red diamond and four other inner circle nodes are 
marked with red stars.
}
\label{fig:vast}
\end{center}
\end{figure}

We note that these conclusions are consistent with the results in 
\cite{MH16}, where an algorithm was proposed to quantify the 
asymmetry caused by the arrow of time.
Our computations also make it clear that, in the Katz setting, use of the
supra-centrality  matrix 
also requires a third parameter to be chosen, for the resolvent system involving
$\MM$.

\section{Conclusions}  \label{sec:conc} 
This work focused on the context where a time-dependent sequence of networks is provided 
for   a given set of nodes. Equivalently, we have an ordered sequence of adjacency matrices, or a three-dimensional tensor. 
It is often useful to  to express  the tensor as a large block matrix, a process known as  
flattening, reshaping, unfolding or matricizing. This corresponds to
representing the system as a single, static network in which nodes make multiple appearances.
Such a representation has the advantage that 
a variety of computational approached can be designed and tested. In particular,
we found that an iterative method based on a regular matrix splitting was particularly effective.  
 However, construction of the block matrix representation    
 must be undertaken with care. In the case of extracting 
resolvent-based network centrality measures that are motivated through the concept of traversals through the network, we highlighted 
a block matrix structure that makes available time-respecting centralities and
illustrated its practical advantages over an alternative formulation.

Interesting avenues for future work in this context include
\begin{itemize}
\item developing strategies for choosing algorithm  parameters, a key example being the length of the time windows, where there is a trade-off between dimensionality and sparsity,      
\item considering matrix functions other than the resolvent,
\item computing other walk-based centrality measures, such as 
    total communicability \cite{BK13} or hub-authority communicability
   \cite{BEK13}, 
\item studying the flattening issue for more general multilayer networks where time is one dimension out of many.
\end{itemize}

\bigskip

\noindent
\textbf{Acknowledgements}
CF would like to thank the Department of Mathematics and Statistics at the University of 
Strathclyde for hospitality during which this work was initiated.

\end{document}